\documentclass[10pt,envcountsame]{llncs}
\usepackage[utf8]{inputenc}
\usepackage{amsfonts}
\usepackage{amsmath}
\usepackage{amssymb}
\usepackage{array}
\usepackage{clrscode3e}
\usepackage{dsfont}
\usepackage{enumerate}
\usepackage{fancyhdr}
\usepackage{graphicx}
\usepackage{latexsym}
\usepackage{mdwlist}
\usepackage{psfrag} 
\usepackage{textcomp} 
\usepackage{tikz}
\usepackage{mathtools}
\newcounter{aux}

\newcommand{\Rounding}{StarRounding}
\newcommand{\opt}{\ensuremath{\mathrm{OPT}}}

\newcommand{\sol}{\ensuremath{\mathrm{SOL}}}

\newcommand{\SF}{\ensuremath{\mathrm{(SF)}}}
\newcommand{\Derandomized}{DeterministicStarRounding}
\newcommand{\Oh}{\mathit{O}}

\newcommand{\plspA}{\hspace{1.5ex}}
\newcommand{\plspB}{\hspace{2ex}}
\newcommand{\plspC}{\hspace{4ex}}

\newcounter{constraint}
\newenvironment{linearprogram}
                {\noindent\begin{tabular}{l@{\plspC}l@{\plspB}%
                              >{\begin{math}\displaystyle}r<{\end{math}}@{\plspA}%
                              >{\begin{math}\displaystyle}c<{\end{math}}@{\plspA}%
                              >{\begin{math}\displaystyle}l<{\end{math}}@{\plspB}%
                              >{\begin{math}\displaystyle}l<{\end{math}}@{\plspB}%
                              l@{}}}
                {\end{tabular}\noindent}%

\newcommand{\profit}{\ensuremath{\mathrm{profit}}}
\newcommand{\E}{\ensuremath{\mathbb{E}}}

\begin{document}
\mainmatter
    \title{Approximation Algorithms for\\
    the Max-Buying Problem with Limited Supply}
    \author{Cristina G. Fernandes%
    \thanks{Partially supported by CNPq 309657/2009-1 and 475064/2010-0.}%
    \and Rafael C. S. Schouery%
    \thanks{Supported by grant \mbox{2009/00387-7}, São Paulo Research Foundation (FAPESP).}}%
    \institute{Department of Computer Science, University of São Paulo, Brazil\\
    \email{\{cris,schouery\}@ime.usp.br}
    }
 
\maketitle

\begin{abstract}
  We consider the Max-Buying Problem with Limited Supply, in which there are~$n$ items, with $C_i$ copies of each item $i$, and $m$ bidders such that every bidder $b$ has valuation $v_{ib}$ for item~$i$. The goal is to find a pricing~$p$ and an allocation of items to bidders that maximizes the profit, where every item is allocated to at most $C_i$ bidders, every bidder receives at most one item and if a bidder $b$ receives item $i$ then \mbox{$p_i \leq v_{ib}$}. Briest and Krysta presented a $2$-approximation for this problem and Aggarwal et al. presented a $4$-approximation for the Price Ladder variant where the pricing must be non-increasing (that is, $p_1 \geq p_2 \geq \cdots \geq p_n$). We present an $e/(e-1)$-approximation for the Max-Buying Problem with Limited Supply and, for every $\varepsilon > 0$, a $(2+\varepsilon)$-approximation for the Price Ladder variant.
\end{abstract}


\section{Introduction}
One interesting economic problem faced by companies that sell products or provide services to consumers is to choose the price of products or services in order to maximize profit. If prices are high then some consumers will not want to (or will not be able to) buy the product  and if the prices are low, the company could obtain a low profit. This is a vastly studied problem, with different models for different situations and a great diversity of approaches~\cite{OrenSW84,OrenSW87,Sen82,Smith86}. 

One way to address this problem is through the nonparametric approach~\cite{RusmevichientongRG06}, where the company collects the preferences of consumers groups (for example, using a website) and optimizes according to some assumptions on the consumer behavior.

In this scenario, we have $n$ products or services (that we will call items) and there are $m$ consumers (that we will call bidders) in the market. At first, we consider that there is an unlimited supply of every item. The auctioneer (the price setter) wants to assign a price $p_i$ for every item $i$ with the objective of maximizing her/his profit (the sum of the prices of sold items considering multiplicities). For this, the auctioneer gathers information about the valuations of the bidders, that is, the largest amount that a bidder $b$ is willing to pay for item $i$ (denoted by $v_{ib}$). 

First of all, we assume that our market is unit-demand, that is, every bidder desires to buy at most one item. We also assume that a bidder $b$ will only buy an item if it is not too expensive (that is, if $p_i \leq v_{ib}$). We call those items feasible for $b$. If there is no feasible item for $b$ then $b$ does not buy any item.

Three models were introduced by Rusmevichientong et al.~\cite{RusmevichientongRG06}. In the \mbox{\textbf{Min-Buying Problem}}, we consider that every bidder will buy one of the less expensive items that are feasible. In the \textbf{Rank-Buying Problem}, we also know the preference order among the items for every bidder and we consider that a bidder will buy the most preferred feasible item. Finally, in the \mbox{\textbf{Max-Buying Problem}}, we consider that every bidder will buy one of the most expensive items that are feasible. 

There are also some restrictions that one can impose on the problems mentioned above. One of such restrictions is called \textbf{limited supply}, that is, every item has some maximum number of copies that can be sold. Sometimes a company knows (or desires) an ordering in the prices of its products and we can impose a \textbf{price ladder}, that is, consider only pricings that have $p_1 \geq p_2 \geq \cdots \geq p_n$. Finally, one can also focus on \textbf{uniform budgets}, where every bidder $b$ has a set of items $I_b$ and a value $V_b$ such that, for every item $i$, $v_{ib} = V_b$ if $i \in I_b$ and $v_{ib} = 0$ otherwise.

Rusmevichientong et al.~\cite{RusmevichientongRG06} showed that, if we impose a price ladder, then one can solve the Min-Buying Problem with Uniform Budgets in polynomial time. Later on, Aggarwal et al.~\cite{AggarwalFMZ04} proved several results for these models considering non-uniform budgets.
They presented a polynomial-time approximation scheme for the Max-Buying Problem with a Price Ladder and showed how to reduce the Rank-Buying Problem with a Price Ladder to the Max-Buying Problem with a Price Ladder. They also presented a $4$-approximation algorithm for the Max-Buying Problem with a Price Ladder and Limited Supply. For the case where we do not have a price ladder, Aggarwal et al. presented an \mbox{$e/(e-1)$-approximation} algorithm for the Max-Buying Problem along with a lower bound of $16/15$, a~$\log(m)$-approximation that can be used for the three models and a~$1+\varepsilon$ lower bound for the Min-Buying Problem, for some constant~$\varepsilon$.

Another variant considered by Aggarwal et al.~\cite{AggarwalFMZ04} is the online version of the Max-Buying Problem with Limited Supply where we know the valuation matrix $v$ in advance but we do not know the arrival order of the bidders and we have to choose a pricing. When a bidder arrives, he buys the most expensive feasible item that still has an unsold copy. They proved that, for any fixed pricing, the revenue obtained by any ordering of the bidders is at least $1/2$ of the revenue obtained by an optimal ordering of the bidders. From this, it follows that any~\mbox{$\alpha$-approximation} for the Max-Buying Problem with Limited Supply (with or without a price ladder) is also $2\alpha$-competitive for the online version.

Briest and Krysta~\cite{BriestK07} showed that the Min-Buying Problem is not approximable within $\Oh(\log^{\varepsilon} m$) for some positive constant $\varepsilon$, unless \mbox{$\mathrm{NP} \subseteq \mathrm{DTIME}(n^{\Oh(\log\log n)})$}, within $\Oh(\ell^\varepsilon)$ where $\ell$ is an upper bound on the number of non-zero valuations per bidder, and within $\Oh(n^\varepsilon)$ under slightly stronger assumptions. They also showed that the Max-Buying Problem with a Price Ladder is strong NP-hard and they presented a $2$-approximation algorithm for the Max-Buying Problem with Limited Supply (without a price ladder).

Guruswami et al.~\cite{GuruswamiHKKKM05} studied the Envy-Free Pricing Problem, where a bidder~$b$ must receive an item in the set \mbox{$D_b = \{i \in I: p_i < v_{ib}\}$} that maximizes~$v_{ib} - p_i$. If such set is empty, then $b$ must either receive no item or receive an item such that $p_i = v_{ib}$. The problem was considered in a more general setting where bidders have valuations for bundles of items (like in a combinatorial auction~\cite[Chap.~11]{NisanRTV07}), but their work focus on unit-demand auctions and also on single-minded bidders. This problem is related with the one of finding a Walrasian Equilibrium~\cite{Walras54} and it is, actually, a relaxation of such equilibrium where we allow unsold items to have non-zero pricing. If we exchange the objective to maximizing the social welfare instead of maximizing the auctioneer's profit, then it is possible to use the well-known VCG mechanism~\cite{Vickrey61,Clarke71,Grove73} to solve the problem in polynomial time, with the nice property that the mechanism is also truthful.

Guruswami et al. proved that the Envy-Free Pricing Problem for the unit-demand case is APX-hard and provided an $\Oh(\log n)$-approximation for it. Latter on, Briest~\cite{Briest08} showed that this problem cannot be approximated within~$\Oh(\log^\varepsilon |B|)$ for some $\varepsilon > 0$ if we assume some specific hardness of refuting random 3SAT-instances or approximating the balanced bipartite independent set problem in constant degree graphs. 

\subsection{Our Results}
We present two new approximation algorithms for the Max-Buying Problem with Limited Supply, one for the general case and another for the case where a Price Ladder is given. Both algorithms improve the previously best know approximation ratio for these problems.

For the Max-Buying Problem with Limited Supply (without a price ladder), we present an $e/(e-1)$-approximation improving the previous upper bound of~$2$ by Briest and Krysta~\cite{BriestK07}. (Recall that $e/(e-1) < 1.582$.) Also, this algorithm has the same approximation ratio as the algorithm for the Max-Buying Problem (with unlimited supply) presented by Aggarwal et al.~\cite{AggarwalFMZ04}. Note that unlimited supply is a particular case of our problem where the number of copies of every item is the number of bidders. We believe that the algorithm is interesting by itself: it uses an integer programming formulation with an exponential number of variables to do a probabilistic rounding and also it explores some structure of the problem that could be useful when developing approximations for the other problems previously described. We also show how to find a deterministic algorithm though derandomization of our algorithm using the method of conditional expectations~\cite{ErdosS73,Spencer87}.

For the Max-Buying Problem with Limited Supply and a Price Ladder, we present a family of algorithms parametrized by a positive rational $\varepsilon$  such that the algorithm is polynomial for constant $\varepsilon$ and provides an approximation ratio of~$2+\varepsilon$.

Notice that, using the result presented by Aggarwal et al.~\cite{AggarwalFMZ04}, our first algorithm is $2e/(e-1)$-competitive for the online version of the Max-Buying Problem with Limited Supply and our second algorithm is $(4+\varepsilon)$-competitive for the online version of the Max-Buying Problem with Limited Supply and a Price Ladder.

This paper is organized as follows. In the next section, we present some notation and describe formally the problem that we address. In Sect.~\ref{sec:lc} we present an~$e/(e-1)$-approximation for the Max-Buying Problem with Limited Supply and in Sect.~\ref{sec:lc-pl} we present, for a given positive rational $\varepsilon$, a $(2+\varepsilon)$-approximation for the Max-Buying Problem with Limited Supply and a Price Ladder. Finally, in Sect.~\ref{sec:conclusion} we present our final remarks. In Appendix~\ref{sec:ommited} we present the omitted proofs and in Appendix~\ref{sec:derandomizing} we show how to derandomize the algorithm presented in Sect.~\ref{sec:lc}.

\section{Model and Notation}\label{sec:models-and-notation}
We denote by $B$ the set of bidders and by $I$ the set of items.  

\begin{definition}
  A \emph{valuation} matrix is a non-negative integer matrix $v$ indexed by~$I \times B$. 
\end{definition}

The number $v_{ib}$ represents the value of item $i$ to bidder $b$.

\begin{definition}
  A \emph{pricing} is a non-negative rational vector indexed by $I$.
\end{definition}

\begin{definition}
  An \emph{allocation} is a vector $x$ indexed by $B$ where $x_b$ is the item allocated to bidder $b$. If a bidder $b$ does not receive an item, then $x_b = \emptyset$.
\end{definition}

Note that an allocation is not necessarily a matching, as the same item can be assigned to more than one bidder (each one receives a 
copy of the item).

\begin{definition}
  Given a valuation $v$ and a pricing $p$, we say that item $i$ is \emph{feasible} for bidder $b$ if $v_{ib} \geq p_i$.
\end{definition}

\begin{definition}
    The \emph{Max-Buying-Limited Problem} consists of, given a valuation~$v$ and a vector $C$ indexed by $I$, finding a pricing $p$ and an allocation $x$ that maximizes the auctioneer's profit such that every item is allocated to at most~$C_i$ bidders and every bidder either receives no item or receives a feasible item.
\end{definition}

Notice that, in spite of the name of the problem, we do not demand that a bidder $b$ receives a feasible item $i$ that maximizes $p_i$. That is, we allow bidder $b$ to receive another item (or none at all) if the most expensive item that is feasible to $b$ is sold out (that is, all copies are allocated to other bidders). 

Next we formalize the variant of the problem where we have a Price Ladder.
\begin{definition}
  The \emph{Max-Buying-PL-Limited Problem} is a variant of the \mbox{Max-Buying-Limited Problem} where the prices are restricted to be non-increasing, that is, $p_1 \geq \cdots \geq p_n$, where $n = |I|$.
\end{definition}

\section{An Algorithm for Limited Supply}\label{sec:lc}

Next, we present a new approximation algorithm for the Max-Buying-Limited Problem with a better ratio than the approximation presented by Briest and Krysta~\cite{BriestK07}. This approximation can also be used for the Max-Buying Problem (without limited supply) and has the same ratio than the approximation (for the Max-Buying Problem) presented by Aggarwal et al.~\cite{AggarwalFMZ04}.

First, consider a solution $(x,p)$ of the Max-Buying-Limited Problem, an item~$i$ and the set $S$ of bidders that bought item $i$. Note that, because $i$ is feasible for every bidder in $S$, we have that~$p_i \leq \min\{v_{ib} : b \in S\}$. If $S \neq \emptyset$ and~\mbox{$p_i <  \min\{v_{ib} : b \in S\}$}, then $(x,p)$ cannot be an optimal solution because one could increase the price of $i$ to obtain a strictly better solution. We conclude that we may assume, w.l.o.g., that $p_i = \min\{v_{ib} : b \in S\}$ for every item $i$ that is bought by a set $S$ of bidders.

We will present an IP formulation for the Max-Buying-Limited Problem that is heavily based in the observation above and in the next definitions. From this formulation, we will design a randomized rounding approximation algorithm.

\begin{definition}
  For an item $i$, let $\mathcal{S}(i) = \{(i,S) : S \subseteq B, \, |S| \leq C_i\}$. We call $(i,S) \in \mathcal{S}(i)$ a \emph{star} of $i$ and we denote by $\mathcal{S}$ the set of all stars, that is,~\mbox{$\mathcal{S} = \bigcup_{i \in I} \mathcal{S}(i)$}.  
\end{definition}

Note that, for $S \subseteq B$ and an item $i$, $S \cup \{i\}$ induces a star in the complete bipartite graph where the parts of the bipartition are $I$ and $B$.

\begin{definition}
  Let $(i,S)$ be a star of item $i$. We denote $\min\{v_{ib} : b \in S\}$ by~$P_{(i,S)}$, that is, the price of item $i$ when sold to the set $S$ of bidders.
\end{definition}

Notice that a feasible solution of the Max-Buying-Limited Problem can be seen as a collection of stars, one for each item, with every bidder in at most one star and the price of an item $i$ being $P_{(i,S)}$, where $(i,S)$ is the star selected for item $i$.

Next, we present our formulation, called \SF{}, in which we have a binary vector~$x$ of variables, with $|\mathcal{S}|$ positions, where $x_{(i,S)}$ is equal to $1$ if and only if the set of bidders that receive item $i$ is precisely $S$. The goal is to determine $x$ that 
\begin{center}
\begin{linearprogram}
\SF & \mbox{maximizes} & \multicolumn{3}{l}{$\displaystyle \sum_{(i,S) \in \mathcal{S}} |S|P_{(i,S)}x_{(i,S)}$}\\[1.5ex]
& \mbox{subject to} & \sum_{(i,S) \in \mathcal{S}(i)} x_{(i,S)} & = & 1,                             & \forall i \in I\\
& & \sum_{(i,S) \in \mathcal{S} : b \in S} x_{(i,S)}    & \leq & 1,  & \forall b \in B\\
& & x_{(i,S)} & \in & \{0,1\},                                               & \forall (i,S) \in \mathcal{S}.
\end{linearprogram}
\end{center}

The formulation \SF{} can be seen as a reduction of our problem to the Set Packing Problem~\cite{GareyJ79}. A similar idea was used by Hochbaum~\cite{Hochbaum82} to obtain an~\mbox{$\Oh(\log n)$-approximation} for the Metric Uncapacitated Facility Location Problem by reducing it to the Set Cover Problem~\cite{GareyJ79}. In our case, we can use the weight structure of the sets to obtain a constant factor approximation for our problem.

This formulation can have $\Omega(|I|2^{|B|})$ variables. But, fortunately, it is possible to solve its linear relaxation in polynomial time using a column generation procedure.

\newcounter{pricing}
\setcounter{pricing}{\value{theorem}}
\begin{lemma}\label{lemma:polynomial}
  The linear relaxation of formulation \SF{} can be solved in polynomial time.~\qed
\end{lemma}

We will use this IP formulation to develop an approximation algorithm for the Max-Buying-Limited Problem using probabilistic rounding. Next we present our algorithm.

\begin{codebox}
\Procname{$\proc{\Rounding}(I,B,v)$}
\li Let $x$ be an optimal solution of the linear relaxation of \SF{} for $(I,B,v)$
\li \For every item $i \in I$
\li   \Do Choose a star $S_i \in \mathcal{S}(i)$ with probability $\mathbb{P}(S_i = (i,S)) = x_{(i,S)}$
\li       Set the price of $i$ as $P_{S_i}$
      \End
\li \For every bidder $b \in B$
\li   \Do Let $i$ be the item such that $S_i = (i,S)$ with $b \in S$ and maximum $P_{S_i}$
\li       Sell item $i$ to bidder $b$
\li       \If there is no such item, bidder $b$ does not receive an item
      \End
\end{codebox}

\begin{theorem}
  \proc{\Rounding} is an $\frac{e}{e-1}$-approximation for the Max-Buying-Limited Problem.\label{thm:rounding}
\end{theorem}
\begin{proof}

First, notice that the objective function of \SF{} can be rewritten as~$\sum_{b \in B} \sum_{(i,S) \in \mathcal{S} : b \in S} P_{(i,S)}x_{(i,S)}$ and that this value is an upper bound on the value of an optimal solution of our problem. We will prove that the expected price paid by bidder $b$ is at least $\frac{e-1}{e}\sum_{(i,S) \in \mathcal{S} : b \in S} P_{(i,S)}x_{(i,S)}$, from where the result will follow.

Consider a bidder $b$ and a non-increasing (in $P_{(i,S)}$) ordering of the stars $(i,S)$ with $b \in S$ and $x_{(i,S)} > 0$. Let $k$ be the number of such stars. If $k = 0$, then the result trivially holds. From now on, we assume that $k > 0$. 

We will denote the $\ell$-th star on this ordering simply by $\ell$, its price by $P_\ell$ and its primal variable by $x_\ell$. If the star $\ell$ is $(i,S)$, we define $c(\ell) = i$, that is, $c(\ell)$ is the item of star $\ell$. Also, we define 
$y_\ell = \sum \{x_{\ell'} : \ell' < \ell  \mbox{ and } c(\ell') = c(\ell)\}$. 
Finally, we denote by $E_\ell$ the event in which star $\ell$ was chosen by \proc{\Rounding}.

Let $f(x) = \frac{1 - e^{-x}}{x}$. For some $1 \leq \ell \leq k$, we denote $f(\sum_{i = \ell}^k x_i)$ simply by~$f_{\ell}$. Notice that, using the fact that $1 - x \leq e^{-x}$, we conclude that $f_\ell \leq 1$ for every~\mbox{$1 \leq \ell \leq k$}. 

Let $\profit(b)$ be the profit that we obtain from bidder $b$. Notice that, for~\mbox{$1 \leq \ell \leq k$}, $\mathbb{E}[\profit(b)|\overline{E}_1,\ldots,\overline{E}_{\ell-1}, E_{\ell}] = P_\ell$ because \proc{\Rounding} will allocate $c(\ell)$ (or another item with the same price) to bidder $b$ because $c(\ell)$ is one of the most expensive items that has $b$ in the chosen star. Also, notice that~$\mathbb{P}(E_\ell | \overline{E}_1,\overline{E}_2,\ldots,\overline{E}_{\ell-1}) = \frac{x_\ell}{1 - y_\ell}$, because we choose the star of an item independently of the star chosen for other items.

Using the observations above, we will prove, by induction in $k - \ell$, that \mbox{$\mathbb{E}[\profit(b)|\overline{E}_1,\overline{E}_2,\ldots,\overline{E}_{\ell-1}] \geq f_\ell \sum_{i = \ell}^kP_ix_i$} for every $1 \leq \ell \leq k$. From this, we will conclude that $\mathbb{E}[\profit(b)] \geq f_1 \sum_{i = 1}^kP_ix_i$ and the result will follow.

If $\ell = k$, then $\mathbb{E}[\profit(b)|\overline{E}_1,\ldots,\overline{E}_{\ell-1}] = P_kx_k/(1 - y_k) \geq P_kx_k \geq f_kP_kx_k$. Now, for $\ell < k$, assume the result is valid for $\ell + 1$. We have that
\begin{align*}  
\E[\profit(b) | \overline{E}_1, \ldots, \overline{E}_{\ell-1}] & = 
\E[\profit(b) | \overline{E}_1, \ldots, \overline{E}_{\ell-1}, E_{\ell}]\,\frac{x_\ell}{1 - y_\ell} & \\
& +  \E[\profit(b) | \overline{E}_1, \ldots, \overline{E}_{\ell-1}, \overline{E}_{\ell}]\left(1 - \frac{x_\ell}{1 - y_\ell}\right)\\
& = P_\ell\,\frac{x_{\ell}}{1-y_\ell} + \left(1 - \frac{x_{\ell}}{1-y_\ell}\right)\E[\profit(b) | \overline{E}_1, \ldots, \overline{E}_{\ell}].
\end{align*}

Using the induction hypothesis, we have that 
\begin{align*}
\MoveEqLeft[12] P_\ell\,\frac{x_{\ell}}{1-y_\ell} + \left(1 - \frac{x_{\ell}}{1-y_\ell}\right)\E[\profit(b) | \overline{E}_1, \ldots, \overline{E}_{\ell}]\\
& \geq P_\ell\,\frac{x_{\ell}}{1-y_\ell} + \left(1 - \frac{x_{\ell}}{1-y_\ell}\right)f_{\ell + 1}\sum_{i = \ell + 1}^kP_ix_i\\
& =  \frac{x_\ell}{1 - y_\ell}\left(P_\ell - f_{\ell + 1}\sum_{i = \ell + 1}^kP_ix_i\right) + f_{\ell + 1}\sum_{i = \ell + 1}^kP_ix_i.
\end{align*}

Now, notice that, because $\sum_{i = \ell}^k x_i \leq 1$ and $P_\ell \geq P_i$ for all $i \geq \ell$, we have that $P_\ell \geq P_\ell\sum_{i = \ell}^kx_i \geq \sum_{i = \ell}^kP_ix_i$. Using the fact that $\frac{x_\ell}{1 - y_\ell} \geq x_\ell$, that $f_{\ell + 1} \leq 1$, and that $1 - x \leq e^{-x}$ for $x \in [0,1]$, we have that
\begin{align*}
\MoveEqLeft[10] \frac{x_\ell}{1 - y_\ell}\left(P_\ell - f_{\ell + 1}\sum_{i = \ell + 1}^kP_ix_i\right) + f_{\ell + 1}\sum_{i = \ell + 1}^kP_ix_i\\
& \geq x_\ell\left(P_\ell - f_{\ell + 1}\sum_{i = \ell + 1}^kP_ix_i\right) + f_{\ell + 1}\sum_{i = \ell + 1}^kP_ix_i\\
& \geq (1 - e^{-x_\ell})\left(P_\ell - f_{\ell + 1}\sum_{i = \ell + 1}^kP_ix_i\right) + f_{\ell + 1}\sum_{i = \ell + 1}^kP_ix_i\\
& =  P_\ell - e^{-x_\ell}P_\ell + e^{-x_\ell} f_{\ell + 1}\sum_{i = \ell + 1}^kP_ix_i\\
& =  P_\ell(1 - e^{-x_\ell} - e^{-x_\ell}x_\ell f_{\ell + 1}) + e^{-x_\ell}f_{\ell + 1}\sum_{i = \ell}^kP_ix_i.
\end{align*}

Before we proceed, we need to show that \mbox{$1 - e^{-x_\ell} - e^{-x_\ell}x_\ell f_{\ell + 1} \geq 0$}. Let~\mbox{$h(x) = 1 - e^{-x} - e^{-x}x t$} for some $0 < t \leq 1$. Notice that $h(0) = 0$ and that~\mbox{$h'(x) = e^{-x} + e^{-x}xt -e^{-x}t \geq e^{-x}xt \geq 0$}, that is, $h(x)$ is non-decreasing for non-negative $x$, from which we conclude that $h(x) \geq 0$ for every non-negative~$x$. Combining this with the fact that $P_\ell \geq \frac{\sum_{i = \ell}^kP_ix_i}{\sum_{i = \ell}^kx_i}$ we have that
\begin{align*}
\MoveEqLeft[8] P_\ell(1 - e^{-x_\ell} - e^{-x_\ell}x_\ell f_{\ell + 1}) + e^{-x_\ell}f_{\ell + 1}\sum_{i = \ell}^kP_ix_i\\
& \geq  \frac{\sum_{i = \ell}^kP_ix_i}{\sum_{i = \ell}^kx_i}\left(1 - e^{-x_\ell} - e^{-x_\ell}x_\ell f_{\ell + 1}\right) + e^{-x_\ell}f_{\ell + 1}\sum_{i = \ell}^kP_ix_i\\
& =  \left(1 - e^{-x_\ell} - e^{-x_\ell}x_\ell f_{\ell + 1} + e^{-x_\ell}f_{\ell + 1}\sum_{i = \ell}^kx_i\right) \frac{\sum_{i = \ell}^kP_ix_i}{\sum_{i = \ell}^kx_i}\\
& =  \left(1 - e^{-x_\ell}\left(1 - f_{\ell + 1}\sum_{i = \ell + 1}^kx_i \right)\right) \frac{\sum_{i = \ell}^kP_ix_i}{\sum_{i = \ell}^kx_i}.
\end{align*}

Now, recall that $f_{\ell + 1}\sum_{i = \ell + 1}^k x_i = 1 - e^{-\sum_{i=\ell+1}^kx_i}$, so we conclude that
\[
\left(1 - e^{-x_\ell}\left(1 - f_{\ell + 1}\sum_{i = \ell + 1}^kx_i \right)\right)
\frac{\sum_{i = \ell}^kP_ix_i}{\sum_{i = \ell}^kx_i}
= f_\ell\sum_{i = \ell}^kP_ix_i.
\]

That is, we have that  $\mathbb{E}[\profit(b)|\overline{E}_1,\overline{E}_2,\ldots,\overline{E}_{\ell-1}] \geq f_\ell \sum_{i = \ell}^kP_ix_i$. 
From this and the fact that $f(x) \geq \frac{e - 1}{e}$ for every $0 < x \leq 1$, we conclude that
\[
\mathbb{E}[\profit(b)] \geq \frac{1 - e^{-\sum_{i = 1}^kx_i}}{\sum_{i = 1}^kx_i}\sum_{i = 1}^kP_ix_i \geq \frac{e - 1}{e}\sum_{i = 1}^kP_ix_i
\]
and the result follows.
\qed \end{proof}

Also, it is easy to prove that the analysis is tight, as we do in the next lemma.

\begin{lemma}
  For every $\varepsilon > 0$, there is an instance where the expected profit of a solution found by \proc{\Rounding} is smaller than  $((e-1)/e + \varepsilon)\mathrm{OPT}$, where $\mathrm{OPT}$ is the value of an optimal solution for this instance.
\end{lemma}
\begin{proof}
  Consider this simple instance: we have a set of items $I$ and only one bidder $b$ such that $v_{ib} = 1$, for every $i \in I$. It is easy to see that an optimal solution for this instance has value $1$. It is also clear that an optimal solution of the linear relaxation has value $1$. One of such optimal solutions is $x$ such that $x_{(i,\{b\})} = 1/|I|$ and $x_{(i,\emptyset)} = 1 - 1/|I|$, for every item $i \in I$. Notice that bidder $b$ pays $1$ if any of the stars $(i,\{b\})$ is chosen and pays $0$ (because is unallocated) otherwise. Thus we have
  \[\mathbb{E}[\profit(b)] = 1 - \left(1 - \frac{1}{|I|}\right)^{|I|} \xrightarrow{|I| \rightarrow \infty} 1 - \frac{1}{e}\]
  and the result follows. \qed
\end{proof}

\section{An Algorithm for Limited Supply with a Price Ladder}\label{sec:lc-pl}

In this section we present, for every $\varepsilon > 0$, a $(2+\varepsilon)$-approximation for the \mbox{Max-Buying-PL-Limited} Problem. We use some ideas from the $4$-approximation algorithm for \mbox{Max-Buying-PL-Limited} Problem developed by Aggarwal et al.~\cite{AggarwalFMZ04}, but in a different way, in order to obtain a better approximation ratio.

\begin{definition} 
  For a valuation matrix $v$, a positive rational $\alpha > 1$, a posite integer $t$, and every non-negative integer $k$, let~\mbox{$d_k = \max\{v_{ib}: i \in I, b \in B\}/\alpha^{k}$}. The \emph{\mbox{Max-Buying-PL-Limited-$(\alpha,t)$} Problem} is a variant of the Max-Buying-PL-Limited Problem where every price is chosen from the set $\{d_0, d_1, \dots\}$ and every bidder receives, for every non-negative~$r$, at most one item with price in~$\{d_{rt}, d_{rt+1}, \ldots, d_{(r+1)t-1}\}$.
\end{definition}

Notice that, in the Max-Buying-PL-Limited-$(\alpha,t)$ Problem, a bidder can receive more than one item but, for every multiple $s$ of $t$, it can receive at most one item with price in $\{d_s, d_{s+1}, \ldots, d_{s + t - 1}\}$.

We start proving that we need to consider only a finite number of possible prices.

\newcounter{dlowerbound}
\setcounter{dlowerbound}{\value{theorem}}
\begin{lemma}\label{lemma:dl}
  Consider an instance of the Max-Buying-PL-Limited-$(\alpha,t)$ Problem. There is a non-negative integer $\ell$, whose value is bounded by a function that is linear on the length of the representation of the instance, such that there is an optimal solution where the price of every item is at least $d_{\ell}$.~\qed
\end{lemma}

We will prove this problem can be solved in polynomial time.

\begin{theorem}\label{thm:alphat}
  The Max-Buying-PL-Limited-$(\alpha,t)$ Problem can be solved in polynomial time.
\end{theorem}
\begin{proof}
Consider items $i$ and $j$ such that $i \leq j$ and a non-negative integer $r$. Let~$P(i,j,r)$ denote the maximum profit achievable for selling items~\mbox{$i, i+1, \ldots, j$} with a pricing such that $p_s \in \{d_{rt}, d_{rt +1}, \ldots, d_{(r+1)t-1}\}$ for every item $s$ and~\mbox{$p_i \geq p_{i+1} \geq \cdots \geq p_j$}. Also let $F(j,r)$ be the maximum profit achievable considering only prices between $d_0$ and $d_{(r+1)t - 1}$ and items between $1$ and $j$. We have the following recurrence:
\[
F(j,r) = \left\{
\begin{array}{ll}
0, & \mbox{if } j = 0,\\
P(1,j,r), & \mbox{if } j > 0 \mbox{ and } r = 0,\\
\displaystyle \max_{0 \leq i \leq j} \{F(i,r-1) + P(i+1,j,r)\}, & \mbox{otherwise.}
\end{array}
\right.
\]
Let $\ell$ be as in Lemma~\ref{lemma:dl} and notice that $F(|I|,\ell)$ is the value of an optimal solution for the \mbox{Max-Buying-PL-Limited-$(\alpha,t)$} Problem. For fixed $t$ and $\alpha$, if we can compute $P(i,j,r)$ in polynomial time, then this recurrence can be solved in polynomial time.

In order to compute $P(i,j,r)$ in polynomial time, we can enumerate every possible pricing (there are at most $|I|^t$ such pricings) and construct a bipartite graph $G$ with bipartition sides $\{i,\ldots, j\}$ and $B$, where, for every item $k$ in~$\{i,\ldots,j\}$ and every bidder $b \in B$, we have an edge $\{i,b\} \in E(G)$ of weight~$p_i$ if and only if $v_{ib} \geq p_i$. Then it remains to find a maximum weighted $b$-matching~\cite{BurkardDM09,Kuhn55} on such graph, where every bidder is matched to at most one item and every item $i$ is matched to at most $C_i$ bidders.\qed \end{proof}

We will now establish some relations involving the value of an optimal solution of the Max-Buying-PL-Limited Problem and the value of an optimal solution of the \mbox{Max-Buying-PL-Limited-$(\alpha,t)$} Problem.

\newcounter{alphaloss}
\setcounter{alphaloss}{\value{theorem}}
\begin{lemma}\label{lema1}
  Let \opt{} be the value of an optimal solution for the Max-Buying-PL-Limited Problem and $\opt'$ be the value of an optimal solution for the Max-Buying-PL-Limited-$(\alpha,t)$ Problem. We have that $\opt' \geq \opt/\alpha$.~\qed
\end{lemma}

\begin{lemma}\label{lema2}
  Consider an optimal solution of the Max-Buying-PL-Limited-$(\alpha,t)$ Problem and let $\opt'$ be the value of such solution. There is a solution of the Max-Buying-PL-Limited Problem, that can be computed in polynomial time, of value $\sol$ such that 
  $\sol \geq \frac{\alpha^t - 1}{\alpha^t - 1 + \alpha^{t-1}}\opt'.$
\end{lemma}
\begin{proof}
  Consider an optimal solution of the Max-Buying-PL-Limited-$(\alpha, t)$ Problem. We will construct a solution for the Max-Buying-PL-Limited Problem by assigning to every bidder $b$ the most expensive item bought by $b$ (which can be done in polynomial time).
  
  For a bidder $b$, let $K$ be the set of integers such that $k \in K$ if and only if~$b$ bought an item of price $d_k$ and let $d_i$ denote the most expensive item bought by~$b$. Remember that, for every $r$, a bidder can buy at most one item of price~$d_{rt + k}$ for $0 \leq k < t$. From this we conclude that
  \[\sum_{k \in K}d_k \leq d_i + \sum_{k \in K\setminus\{i\}}d_{\lfloor k/t \rfloor t} \leq d_i + \sum_{r \geq 0}d_{i+rt+1}.\]

Now, notice that $d_{i+rt+1} =  d_i/\alpha^{rt+1}$, from which we conclude that
\[
\sum_{k \in K}d_k 
\leq d_i + \sum_{r \geq 0}d_{i+rt+1}
= d_i \left(1 + \frac{1}{\alpha}\sum_{r \geq 0}\left(\frac{1}{\alpha^t}\right)^{r} \right)
\leq d_i \left(1 + \frac{1}{\alpha}\left(\frac{1}{1 -\frac{1}{\alpha^t}}\right)\right).
\]

Notice that the profit obtained from $b$ in the solution of the Max-Buying-PL-Limited-$(\alpha,t)$ is exactly $\sum_{k \in K}d_k$ and the profit obtained from $b$ in the solution found for the \mbox{Max-Buying-PL-Limited} Problem is exactly $d_i$. But we concluded that $d_i \geq \frac{\alpha^t - 1}{\alpha^t - 1 + \alpha^{t-1}}\sum_{k \in K}d_k$, so the result follows.
\qed \end{proof}

Combining the results of Lemmas~\ref{lema1} and~\ref{lema2}, we can derive an approximation for the Max-Buying-PL-Limited Problem.

\begin{corollary}\label{cor:pllc}
  For every positive integer $t$ and rational $\alpha > 1$, there is an~\mbox{$\frac{\alpha(\alpha^t - 1 + \alpha^{t-1})}{\alpha^t - 1}$}-approximation for the Max-Buying-PL-Limited Problem.
\end{corollary}  
\begin{proof}
  The algorithm is very simple: find an optimal solution $(x,p)$ of the \mbox{Max-Buying-PL-Limited-$(\alpha,t)$} Problem as described in Theorem~\ref{thm:alphat} and return~$(\tilde{x}, p)$, where the item allocated to a bidder $b$ in $\tilde{x}$ is the most expensive item allocated to $b$ in $(x,p)$.

  Let $\sol$ denote the value of the solution found,  $\opt$ denote the value of an optimal solution of the \mbox{Max-Buying-PL-Limited} Problem and $\opt'$ denote the value of an optimal solution of the \mbox{Max-Buying-PL-Limited-$(\alpha,t)$} Problem.

  By Lemma~\ref{lema1} we have that $\opt'  \geq \opt/\alpha$ and by Lemma~\ref{lema2} we have that \mbox{$\sol \geq \frac{\alpha^t - 1}{\alpha^t - 1 + \alpha^{t-1}}\opt'$}, hence we conclude that $\sol \geq \frac{\alpha^t - 1}{\alpha(\alpha^t - 1 + \alpha^{t-1})}\opt$ and we obtain the desired approximation ratio. \qed 
\end{proof}

\begin{corollary}\label{cor:2e}
  For every $0 < \varepsilon < 1$, there is a $(2+\varepsilon)$-approximation of the Max-Buying-PL-Limited Problem.
\end{corollary}
\begin{proof}
  Let $\alpha = 1 + \frac{\varepsilon}{2}$ and $t = \lceil \log_{\alpha} (\frac{2}{\varepsilon} + 1)\rceil$ (notice that, because $\varepsilon < 1$, we have that $t$ is a positive integer). We have that
  \[
  \frac{\alpha(\alpha^t - 1 + \alpha^{t-1})}{\alpha^t - 1} =
  1 + \alpha + \frac{1}{\alpha^t - 1} \leq 
  1 + \left(1+\frac{\varepsilon}{2}\right) + \frac{\varepsilon}{2} =
  2 + \varepsilon.
  \]
  That is, we only have to carefully choose $\alpha$ and $t$ and use the algorithm from Corollary~\ref{cor:pllc}.\qed
\end{proof}

\section{Final Remarks}\label{sec:conclusion} 

In this paper we focused on the Max-Buying Problem when we have limited supply, considering the case with the price ladder restriction and without this restriction. 

Our results improve the previously best known approximation ratios for both problems (with and without the Price Ladder restriction). This technique of enumerating all possible allocations of items to bidders might also help in other pricing problems.

We believe that pricing problems with limited supply are very interesting because this is a realistic restriction and also a hard one to be considered from the approximation algorithm perspective. Even though in general the Max-Buying Problem seems to be simpler than the Min-Buying Problem, the Rank-Buying Problem, and the Envy-Free Pricing Problem, it is not trivial to develop good approximations for it when we have limited supply.

There are also some open problems. It is interesting to notice that our algorithm for the Max-Buying-Limited Problem has the same ratio as the algorithm presented by Aggarwal et al.~\cite{AggarwalFMZ04} for the Max-Buying Problem. It would be nice to develop an approximation with ratio better than $e/(e-1)$ for the Max-Buying Problem (and if possible also for the Max-Buying-Limited Problem) or to prove that this value is a lower bound on the approximation ratio of every algorithm for this problem.

In the case of the Max-Buying-PL-Limited Problem, it would be interesting to design a PTAS (since there is a PTAS for the unlimited case, developed by Aggarwal et al.~\cite{AggarwalFMZ04}) or to prove that the problem is $\mathrm{APX}$-hard. Also, notice that the Price Ladder is not helping to achieve a better approximation ratio as it happens for the unlimited supply version. This is somehow against our intuition that knowing the prices order would make it easier to find good pricings. We do not know if this is something intrinsic to this problem or if there are other ways to exploit the Price Ladder in order to obtain better approximations. 

\bibliographystyle{splncs03} 
\bibliography{auctionbib.bib} 
\newpage 
\appendix 

\section{Omitted Proofs}\label{sec:ommited}
In this appendix, we present the proofs that were omitted due to space constraints.

\setcounter{aux}{\value{theorem}}
\setcounter{theorem}{\value{pricing}}
\begin{lemma}
  The linear relaxation of formulation \SF{} can be solved in polynomial time.
\end{lemma}
\setcounter{theorem}{\value{aux}}
\begin{proof}
  First, notice that the dual of the linear relaxation of \SF{} is 
\begin{center}
\begin{linearprogram}
&   \mbox{minimize} & \multicolumn{3}{l}{$\displaystyle \sum_{i \in I} \alpha_{i} + \sum_{b \in B} \beta_{b}$}\\[1.5ex]
&   \mbox{subject to}  & \alpha_i + \sum_{b \in S} \beta_{b} &\geq &P_{(i,S)}|S|  & \forall (i,S) \in \mathcal{S}\\
& & \alpha_i & \in & \mathbb{Q} & \forall i \in I\\
& & \beta_b & \geq & 0 & \forall b \in B.
\end{linearprogram}
\end{center}

We will prove that we can solve the separation problem for this dual linear program in polynomial time. From this, and using the result of Gr{\"o}tschel et al.~\cite{GrotschelLS88}, we will conclude that we can solve this linear program in polynomial time (and also obtain an optimal solution of the primal).

Consider that one is given a vector $(\alpha, \beta)$ that is a candidate to be a solution of the dual above. We want, in polynomial time, to decide if $(\alpha, \beta)$ is indeed a solution and, if $(\alpha, \beta)$ is not a solution, to provide a violated inequality.

First of all, we will assume that, for every bidder $b$, we have that $\beta_b \geq 0$, because if there is a bidder $b$ such that $\beta_b < 0$, then $(\alpha, \beta)$ is clearly not feasible and we can return the inequality $\beta_b \geq 0$ to prove it. Also, we will assume, w.l.o.g.\ that $C_i \leq |B|$ for every item $i$.

Now, for some $i \in I$, $b \in B$, and $k \in \{1,\ldots,C_i\}$, let $\mathcal{S}(i,b,k)$ denote the set of stars $(i,S)$ such that $P_{(i,S)} = v_{ib}$, $b \in S$, and $|S| = k$. Note that, if $\mathcal{S}(i,b,k) \neq \emptyset$, then there is a star $(i,S) \in \mathcal{S}(i,b,k)$ that minimizes $\sum_{b \in S} \beta_b$. Also, notice that there is a star $(i,S') \in \mathcal{S}(i,b,k)$ such that $\alpha_i + \sum_{b \in S'} \beta_b < P_{(i,S')}|S'|$ if and only if $\alpha_i + \sum_{b \in S} \beta_b < P_{(i,S)}|S|$ because $P_{(i,S)} = P_{(i,S')} = v_{ib}$ and $|S| = |S'| = k$.

We state that such $(i,S) \in \mathcal{S}(i,b,k)$ that minimizes $\sum_{b \in S} \beta_b$ can be found (if it exists) in polynomial time. One has only to consider the first $k-1$ bidders~$b' \neq b$ in non-increasing order of $\beta_{b'}$ such that $v_{ib'} \geq v_{ib}$  (it is clear that if there is less than $k-1$ such bidders, then $\mathcal{S}(i,b,k)$ is empty). Also, we can decide if~\mbox{$\alpha_i + \sum_{b \in S} \beta_b \geq P_{(i,S)}|S|$} in polynomial time .

To conclude our proof, notice that
\[\mathcal{S} = \bigcup_{\substack{i \in I, b \in B\\ k \in \{1,\ldots,C_i\}}} \mathcal{S}(i,b,k).\]

That is, one can iterate over all $i \in I$, $b \in B$, and $k \in \{1,\ldots,C_i\}$ and decide if there is a star in $\mathcal{S}(i,b,k)$ that violates an inequality, thus our proof follows.
\qed \end{proof}

\setcounter{aux}{\value{theorem}}
\setcounter{theorem}{\value{dlowerbound}}
\begin{lemma}
  Consider an instance of the Max-Buying-PL-Limited-$(\alpha,t)$ Problem. There is a non-negative integer $\ell$, whose value is bounded by a function that is linear on the length of the representation of the instance, such that there is an optimal solution where the price of every item is at least $d_{\ell}$.
\end{lemma}
\setcounter{theorem}{\value{aux}}

\begin{proof}
  Let $k$ be the length of the binary representation of the instance and let $\ell$ be such that $d_\ell < 1 \leq d_{\ell - 1}$. We will prove that $\ell$ is linear on $k$ and that there is an optimal solution whose price of every item is at least $d_{\ell}$.

  By the definition of $d_{\ell-1}$, we have that $\alpha^{\ell - 1} \leq \max\{v_{ib} : i \in I, b \in B\}$. Also, notice that $\max\{v_{ib} : i \in I, b \in B\} \leq 2^k$, that is, $\alpha^{\ell - 1} \leq 2^k$, from where we conclude that $(\ell - 1)\log_{2}\alpha \leq k$, and deduce that \[\ell \leq \frac{k}{\log_{2}\alpha} + 1 = \Oh(k).\]

  Consider now an optimal solution $(x,p)$ of the Max-Buying-PL-Limited-$(\alpha,t)$ Problem and suppose that there is at least one item $i$ of price $p_i < d_{\ell}$. Also, suppose, w.l.o.g., that items with price zero are not allocated in $(x,p)$, since they do not contribute for the value of the solution. We can construct another solution~$(x,\tilde{p})$ by setting $\tilde{p}_i = d_{\ell}$ for every item $i$ of price $p_i < d_{\ell}$ and $\tilde{p}_i = p_i$ for every item $i$ with $p_i \geq d_{\ell}$. Notice that $(x,\tilde{p})$ is a feasible solution with the same profit as $(x,p)$, because it respects the price ladder and $x$ allocates feasible items to bidders. We conclude that there is an optimal solution where the price of every item is at least $d_{\ell}$. 
\qed \end{proof}

\setcounter{aux}{\value{theorem}}
\setcounter{theorem}{\value{alphaloss}}
\begin{lemma}
Let \opt{} be the value of an optimal solution for the Max-Buying-PL-Limited Problem and let $\opt'$ be the value of an optimal solution for the Max-Buying-PL-Limited-$(\alpha,t)$ Problem. We have that $\opt' \geq \opt/\alpha$.
\end{lemma}
\setcounter{theorem}{\value{aux}}

\begin{proof}
  Consider an optimal solution of the Max-Buying-PL-Limited Problem. By rounding down the price of every item to the nearest $d_k$, we obtain a feasible solution for the Max-Buying-PL-Limited-$(\alpha, t)$ Problem. Notice that~\mbox{$d_{k-1} = \alpha d_k$}, that is, we lose at most a factor of $\alpha$ by rounding down. An optimal solution for the Max-Buying-PL-Limited-$(\alpha,t)$ Problem has value not smaller than the value of this solution, thus the result follows.
\qed \end{proof}

\section{Derandomizing}\label{sec:derandomizing}
In this appendix we show how to design a deterministic $e/(e-1)$-approximation for the Max-Buying-Limited Problem using a derandomization of the algorithm \proc{\Rounding} presented in Sect.~\ref{sec:lc}. We start with some useful definitions.

\begin{definition}
  We denote by $\profit$ the random variable that represents the value of the solution found by \proc{\Rounding}. Also, for a bidder $b$, we denote by $\profit(b)$ the random variable that represents the profit obtained by \mbox{\proc{\Rounding}} from bidder $b$.
\end{definition}

\begin{definition}
  Let $S_1,\ldots,S_N$ be a non-increasing ordering in $P_s$ of the stars with $x_S > 0$. We denote by $E_j$ the event in which \proc{\Rounding} selects star~$S_j$ and by $\overline{E}_j$ the event in which $S_j$ is not selected by \proc{\Rounding}.
\end{definition}

Our algorithm uses conditional expectations~\cite{ErdosS73,Spencer87} to decide if it should keep a specific star in the solution. For this we will have to make a decision for a star~$k$, that will be denoted by $d_k$, as shown below.

\begin{definition}
  Let $S_1,\ldots,S_N$ be a non-increasing ordering in $P_s$ of the stars with $x_S > 0$, let $K \subseteq \{1,\ldots,N\}$ and, for every $k \in K$, let $d_k \in \{E_k, \overline{E}_k\}$. We say that $\{d_k\}_{k \in K}$ is \emph{feasible} if, for every item $i \in I$, there is at most one star $\ell$ of $i$ such that $\ell \in K$ and $d_\ell = E_\ell$.
\end{definition}

That is, we say that $\{d_k\}_{k \in K}$ is feasible if there is at most one chosen star for every item.

\begin{lemma}\label{lemma:expectation}
    Let $S_1,\ldots,S_N$ be a non-increasing ordering in $P_s$ of the stars with $x_S > 0$. Let $b$ be a bidder,
    let $K \subseteq N$ be such that if $\ell \in K$ contains $b$ then
    for every star $\ell' < \ell$ that contains $b$ we have that $\ell' \in K$. It is possible to calculate
    $\mathbb{E}[\profit(b)|\{d_k\}_{k \in K}]$ in polynomial time.
\end{lemma}

\begin{proof}
  First, notice that $N$ is polynomial in $|I| + |B|$.
  If there is $k \in K$ such that $d_k = E_k$ and $S_k$ contains bidder $b$, then because we are considering the stars in a non-increasing order and \proc{\Rounding} is such that a bidder gets an item with maximum price among the ones that selected it, we have that 
  $\mathbb{E}[\profit(b)|\{d_k\}_{k \in K}] = P_{S_{k^*}}$ 
  where $S_{k^*}$ is the star with smallest index that contains $b$ and $d_{k^*} = E_{k^*}$.

  If this is not the case, then let $T_1, T_2, \ldots, T_t$ be the subsequence of $S_1, \ldots, S_N$ of stars that contain $b$ and let $f$ be the mapping of indexes of $T_1,\ldots,T_t$ to $S_1,\ldots,S_N$ (that is, $T_r = S_{f(r)}$). If $f(t) \in K$, then every star that contains~$b$ is in $\{S_k\}_{k \in K}$ and none of them is selected, from which we conclude that $\mathbb{E}[\profit(b)|\{d_k\}_{k \in K}] = 0$. If $f(t) \notin K$, let $T_s$ be the first star (in the ordering) that is not contained in $\{S_k\}_{k \in K}$ (notice that, by hypothesis, $f(r) \in K$ and \mbox{$d_{f(r)} = \overline{E}_{f(r)}$} for every $T_r$ with \mbox{$1 \leq r < s$}). We have that
  \begin{align*}
    \mathbb{E}[\profit(b)|\{d_k\}_{k \in K}] 
    &= \mathbb{E}[\profit(b)|\{d_k\}_{k \in K},E_{f(k)}]\mathbb{P}(E_{f(s)}|\{d_k\}_{k \in K})\\
    &+ \mathbb{E}[\profit(b)|\{d_k\}_{k \in K},\overline{E}_{f(s)}]\mathbb{P}(\overline{E}_{f(s)}|\{d_k\}_{k \in K})\\
    & = P_{f(s)}\mathbb{P}(E_{f(s)}|\{d_k\}_{k \in K})\\ 
    & + \mathbb{E}[\profit(b)|\{d_k\}_{k \in K},\overline{E}_{f(s)}]\mathbb{P}(\overline{E}_{f(s)}|\{d_k\}_{k \in K}).
  \end{align*}

  Notice that $\mathbb{P}(E_{f(s)}|\{d_k\}_{k \in K}) = 0$ if there is a star $S_u$ with $u \in K$ such that $d_s = E_s$
  and $S_u$ and $T_s$ are stars of the same item (that is, this item already has a selected star, so $T_k$ does not have a chance of being selected). Otherwise, $\mathbb{P}(E_{f(s)}|\{d_k\}_{k \in K}) = x_{f(s)}/(1 - y)$ where \mbox{$y = \sum_{s \in K}\{x_s| c(S_s) = c(T_k)\}$}. 

  Finally, using the fact that \mbox{$\mathbb{E}[\profit(b)|\{d_k\}_{k \in K},\overline{E}_{f(s)}, \overline{E}_{f(s+1)},\ldots, \overline{E}_{f(t)}] = 0$}, we can calculate $\mathbb{E}[\profit(b)|\{d_k\}_{k \in K},\overline{E}_{f(s)}]$ in $t-s+1$ steps.\qed
\end{proof}

\begin{corollary}\label{cor:expectation}
  Let $S_1,\ldots,S_N$ be a non-increasing ordering in $P_s$ of the stars with $x_S > 0$. For any $0 \leq j \leq N$ and every feasible $\{d_1, d_2, \ldots, d_j\}$, it is possible to calculate $\mathbb{E}[\profit|d_1,\ldots,d_j]$ in polynomial time.
\end{corollary}
\begin{proof}
  Notice that $\mathbb{E}[\profit|d_1,\ldots,d_j] = \sum_{b \in B}\mathbb{E}[\profit(b)|d_1,\ldots,d_j]$ and also notice that $\{d_1, \ldots, d_j\}$ satisfies the property that if a star $\ell \in \{1,\ldots,j\}$ contains~$b$ then for every star $\ell' < \ell$ that contains $b$ we have that $\ell' \in \{1,\ldots,j\}$. We can use Lemma~\ref{lemma:expectation} to calculate the expected value for every bidder and compute the summation.
\end{proof}

We are now ready to present the derandomized algorithm, called \mbox{\proc{\Derandomized}}.

\begin{codebox}
\Procname{$\proc{\Derandomized}(I,B,v)$}
\li Let $x$ be an optimal solution of the linear relaxation for the input $(I,B,v)$
\li Let $S_1,\ldots,S_N$ be a non-increasing ordering in $P_s$ of the stars with $x_S > 0$
\li \For $j \gets 1$ \To $N$
\li \Do Let $i$ be the item of star $S_j$
  \li \If $S^*_i$ is defined or $\mathbb{E}[\profit|d_1,\ldots,d_{j-1},E_j] < \mathbb{E}[\profit|d_1,\ldots,d_{j-1},\overline{E}_j]$
    \li \Then $d_j \gets \overline{E}_j$
    \li \Else $S^*_{i} \gets S_j$
    \li       $d_j \gets E_j$
    \End
\End
\li \For every bidder $b \in B$
\li   \Do Let $i$ be the item such that $S^*_i = (i,S)$ with $b \in S$ and maximum $P_{S_i}$
\li       Sell item $i$ to bidder $b$
\li       \If there is no such item, bidder $b$ does not receive an item
      \End
\end{codebox}

\begin{theorem}
  \proc{\Derandomized} is an $e/(e-1)$-approximation for the Max-Buying-Limited Problem.
\end{theorem}
\begin{proof}
  We only have to prove that this derandomization provides the same approximation ratio as before and that \proc{\Derandomized} is a polynomial-time algorithm. Notice that, for every $1 \leq j \leq N$, we have that
  \begin{align*}
    \mathbb{E}[\profit|d_1,d_2,\ldots,d_{j-1}] 
    & = \mathbb{E}[\profit|d_1,d_2,\ldots,d_{j-1},E_i]\mathbb{P}(E_i| d_1,d_2,\ldots,d_{j-1})\\
    & + \mathbb{E}[\profit|d_1,d_2,\ldots,d_{j-1},\overline{E}_i]\mathbb{P}(\overline{E}_i|d_1,d_2,\ldots,d_{j-1})\\
    & \leq \mathbb{E}[\profit|d_1,d_2,\ldots,d_{j-1},d_j]
  \end{align*}
  where $d_j$ is the decision that we made on step $j$. From this, using induction and by Theorem~\ref{thm:rounding}, we conclude that  $\mathbb{E}[\profit|d_1,\ldots,d_N] \geq \mathbb{E}[\profit] \geq e/(e-1)$. 

  For the time consumption analysis, recall that $N$ is polynomial in \mbox{$|I|+|B|$} and that, by Corollary~\ref{cor:expectation}, we can compute $\mathbb{E}[\profit|d_1,\ldots,d_{j-1},E_j]$ and $\mathbb{E}[\profit|d_1,\ldots,d_{j-1},\overline{E}_j]$ in polynomial time.\qed
\end{proof}
\end{document}